\documentclass{aims}
\usepackage{amsmath}
  \usepackage{paralist}
\usepackage{graphicx}  \usepackage{multirow}
\usepackage{multicol}
\usepackage{epstopdf}
 \usepackage[colorlinks=true]{hyperref}
\hypersetup{urlcolor=blue, citecolor=red}

\def\currentvolume{X}
 \def\currentissue{X}
  \def\currentyear{200X}
   \def\currentmonth{XX}
    \def\ppages{X--XX}

\usepackage{amsmath}
\usepackage{amssymb}
\usepackage{amsthm}
\usepackage{mathtools}
\usepackage{latexsym}
\usepackage{textcomp}

\usepackage{tikz}
	\usetikzlibrary{positioning,chains,fit,shapes,calc,graphs}
\usepackage{stackrel}

\newtheorem{theorem}{Theorem}

\newtheorem{lemma}{Lemma}
\newtheorem{prop}{Proposition}
\newtheorem{remark}{Remark}
\newtheorem{conjecture}{Conjecture}
\newtheorem{example}{Example}

\definecolor{myblue}{RGB}{80,80,160}
\definecolor{mygreen}{RGB}{80,160,80}

\title[Optimal Colorings of Max $k$-Cut Game]{Optimal Colorings of Max $k$-Cut Game}

\author[Andrea Garuglieri Dario Madeo Chiara Mocenni Giulia Palma and Simone Rinaldi]{}

\email{garuglieri(at)student.unisi.it}
\email{dario.madeo(at)unisi.it}
\email{chiara.mocenni(at)unisi.it}
\email{giulia.palma2(at)unisi.it}
\email{rinaldi(at)unisi.it}

\thanks{$^*$ Corresponding author: Chiara Mocenni}

\begin{document} 
	

	
	\maketitle
	
\centerline{\scshape Andrea Garuglieri, Dario Madeo, Chiara Mocenni$^*$, Giulia Palma and Simone Rinaldi}
\medskip
{\footnotesize
 \centerline{Department of Information Engineering and Mathematics}
   \centerline{University of Siena}
   \centerline{Via Roma 56, 53100, Siena, Italy}
} 

\bigskip

 \centerline{(Communicated by the associate editor name)}	
	
	\begin{abstract}
		We investigate {\it strong Nash equilibria} in the {\it max $k$-cut game} on an undirected and unweighted graph with a set of $k$ colors, where vertices represent players and the edges indicate their relations. Each player $v$ chooses one of the available colors as its own strategy, and its payoff (or utility) is the number of neighbors of $v$ that has chosen a different color. Such games are significant since the model loads of real-worlds scenario with selfish agents and, moreover, they are related to fundamental classes of games.
		Few results are known concerning the existence of strong equilibria in max $k$-cut games in this direction. In this paper we make some progress in the understanding the properties of strong equilibria.  
		In particular, our main result is to show that optimal solutions are
		7-strong equilibria. This means that in order a coalition of nodes is able to deviate and drive the system towards a different  configuration, i.e. a different coloring, a number of nodes of the coalition strictly larger than 7 is necessary. We also conjecture that in a generic graph with $n$ nodes, any optimal coloring is also a $n$-strong equilibrium. Furthermore, we prove some properties of minimal subsets with respect to a strong deviation, showing that each of their nodes will deviate towards
		the color of one of their neighbors.
	\end{abstract}

	\section{Introduction}
	
The max $k$-cut problem consists in assigning colors, or applying a cut, to the nodes of a graph in a way that neighbors nodes have possibly different colors, thus ensuring the highest heterogeneity of colors in the graph. 

 The max $k$-cut problem has been extensively studied due to the great interest for its real-world applications. For example, consider radio towers as players. Their goal is selecting a frequency different from the ones of neighboring radio-towers, because in this way they minimize the interference. Another possible scenario is given by a set of companies that have to decide which product to manufacture in order to maximize their revenue, and to do so they have to minimize the number of competitors (for example the ones that are in the same geographic area) that produce the same product. 

In \cite{KARP} it has been shown that finding an optimal solution for it is $NP$-complete. 
Modern approaches use approximation algorithms \cite{frieza,theta} or  heuristic-based ones \cite{QUBO,memes} in order to find a good enough solution, trading accuracy for computational time.	
See \cite{comput} for a quick excursus on the subject and \cite{vilmar} for a more thorough exploration.

The \textit{max k-cut game} constitutes a strategic version of the max k-cut problem. One of the main problems for the max $k-$cut game is to prove the existence of \textit{Strong Nash Equilibria} (briefly, \textit{SE}) \cite{aumann-se-2}. Strong equilibrium corresponds to have colorings in which no coalition, assuming the actions of its complements as given, can cooperatively deviate in a way that benefits all of its members, i.e. every player of the coalition strictly improves its utility. Indeed, in this line of research the main focus is on equilibrium concepts which are resilient to deviations of groups. From this point of view, the concept of Strong Equilibrium is relevant. A weaker version is the \textit{q-Strong Equilibrium} (briefly, \textit{q-SE}), for some $q \leq n$, in which only coalitions of at most $q$ players are allowed to cooperatively change their strategies. 
Trivially, the $1-SE$ is equivalent to the \textit{Nash equilibrium} (briefly, \textit{NE}), and the $n-SE$ is equivalent to the $SE$.
Concerning the existence of strong equilibria in max k-cut games, not much is known. 
In \cite{gourves-monnot-1} it has been proved that an \textit{optimal strategy profile} (also called \textit{optimal coloring}), i.e. a coloring that maximizes the sum of the players’ utilities (equivalently, a coloring that maximizes the $k-$cut), is an SE for the max $2-$cut game. Then, the authors showed that it is a 3-SE, for the max $k-$cut game, for any $k \geq 2$, but an optimal strategy profile is not necessarily a 4-SE, for any $k \geq 3$. 

Later, in \cite{gourves-monnot-2} it has been shown that, if the number of colors is at least the number of players minus two, then an optimal strategy profile is an SE. 
Lastly, they showed that the dynamics, in the event that at each step a coalition can deviate in such a way that all of its members strictly improve their utility by changing strategy, can cycle. 
An immediate consequence of the above results is that it seems hard to understand whether a SE always exists for max $k-$cut games. 
The authors in \cite{gourves-monnot-2} conjecture that a SE always exists for the \textit{max k-cut game}. It is conjectured that any optimal coloring is an SE. The present paper is a step ahead in this direction.

The most important existing result has been provided in \cite{mt}. In particular, the authors showed that on undirected unweighted graphs, optimal colorings are 5-Strong Equilibria (5-SE), i.e. colorings in which no coalition of at most 5 vertices can profitably deviate.
In \cite{smorodinskis} this game has been extended to hypergraphs, by examining two possible extensions of the payoff function.

In this paper we consider a \textit{max k-cut game} played by $n$ individuals or players. The individuals are assumed to be arranged on an undirected and unweighted graph; specifically, nodes of the graph represent the individuals, while the edges describe the connections among them. Each edge has its own weight represented by a positive real number. The strategy space of each player is composed by a set of $k$ available colors (i.e. $\{1, \dots, k \}$). For the sake of simplicity, we assume that the color set is the same for each player. Given a strategy profile or a coloring 
(i.e.  the sequence of colors chosen by players), the utility (or payoff) of a player $g$ is the sum of the weights of edges $\{g,v\}$ incident to $g$, such that the color chosen by $g$ is different from the one chosen by $v$.
When all weights associated to links are equal to $1$, then the payoff corresponds to the number of neighbors with color different from its color. Each player is selfish, and then its objective is to maximize its own utility. 
The main problem concerning a max $k$-cut game is related to the possibility for players to achieve autonomously a social optimum (i.e. maximize the cut value by themselves) rather than forcing individuals by an external regulator. In other words, the max $k$-cut problem is to partition the $n$ nodes of a graph in $k$ subsets such that the sum of the weights of the edges that connect nodes belonging to the same subset is minimized.
Indeed, in such games on graphs it is beneficial for players to anti-coordinate their choices with the ones of their neighbors by selecting different colors. Therefore, the players may attempt to increase their utility by coordinating their choices in groups, called \textit{coalitions}.  


\paragraph{Our results}
In our work, we extend the main results of \cite{mt}, by showing important properties of minimal subsets with respect to a strong deviation and proving that optimal colorings are 7-{\it SE}. 

This fact has many relevant conceptual consequences and applications. Indeed, we extend previous results by showing that optimal colorings are robust and resilient against larger coalitions aimed at forcing groups of players to selfishly diverge from the optimal equilibrium. We also provide a conjecture that, whether it will be proved true, ensures the same robustness and resilience with respect to even larger dimensions of the coalitions.

The article is structured as follows.
First, in the next section, we give the definitions useful for our study, we point out some relevant results and we introduce the main problem.
Section 3 is devoted to some properties of minimal subsets with respect to a strong deviation and to an alternative formula for computing the cut value increase due to a deviation, along with an extension to undirected weighted graphs. 
Then, in Section 4, we prove that optimal colorings are 7-Strong equilibria by exploiting the results of the previous section. 
We conclude the article dealing with issues stemming from the approach used to prove our main results and pointing out some open problems, in  Section 5.


\section{Preliminaries}
Let $G=\langle V, E\rangle$ be an undirected unweighted graph, where $V$ is a finite set of $n$ vertices, and $E$ is the set of $m$ edges, i.e. a collection of couples of $V$. We assume that no self-loops are present, and that, given two nodes, there is at most one edge connecting them.
The graph $G$ can be also represented by the adjacency matrix $A = \{a_{v,j}\}$, where $a_{v,j} = 1$ if there is an edge connecting $v$ and $j$. The adjacency matrix of an undirected graph is symmetric.

For $v\in V$, the \textit{degree} of the vertex $v$ is $\displaystyle\delta(v) =
			\big|\{j\in V: \{v,j\}\in E\}\big|$.

Given $G$ and a set of colors $K = \{1,\dots, k\}$, 	the degree of $v$ in the coloring $\sigma$ with respect to the color $a$ is $\delta(v,\sigma, a):=\big|\{j\in V: \{v,j\}\in E, \sigma_j=a\}\big|$.		
		
The \textit{max k-cut problem} consists of partitioning the vertices of $G$ into $k$ subsets, denoted by $V_1, \dots, V_k$, such that the number of nodes having neighbors in different sets is maximized. The \textit{max k-cut game} constitutes a strategic version of the max k-cut problem, as it is defined below. There are $|V|$ players, each of which have the same \textit{strategy set}, that is the set of colors $K$.

A \textit{strategy profile}, also called \textit{coloring}, denoted by $\sigma \in K^n$, is a labeling of vertices of $G$ in which a color $\sigma_v$ is assigned to each player $v$. A strategy profile can be seen as a vector $\sigma \in K^n$ containing the strategies chosen by each vertex.

Given a graph $G$ and a coloring $\sigma \in K^n$, first we define \textit{cut of a graph} $G$ the set $E(\sigma):= \lbrace\{i,j\}\in E:\ \sigma_i\neq\sigma_j\rbrace$. The cardinality of $E(\sigma)$ is called \textit{cut value}, or \textit{size of the cut} and it is denoted by $S(\sigma)$.
The \textit{cut difference} between the coloring $\gamma$ and another coloring $\sigma$ is $\Delta S(\sigma,\gamma):=S(\gamma)-S(\sigma)$.
Referring to the coloring $\sigma$, for each $v\in V$, the \textit{utility}, or \textit{payoff}, of player $v$ is defined as $\mu_v(\sigma)= \big|\lbrace u\in V:\
					\{u,v\}\in E,\ \sigma_u\neq\sigma_v\rbrace\big|$.
Using the formalization introduced by the adjacency matrix $A$, we can also write the payoff as: \begin{equation}\label{eqn:payoff_adj}\mu_v(\sigma) = \sum_{\substack{u \in V \\ \sigma_u \neq \sigma_v}} a_{v,u}.
\end{equation}
Thanks to symmetry of $A$, we have that, given a coloring $\sigma$ and two disjoint subsets of the nodes $V_1, V_2 \subseteq V$, we have that:

\begin{equation}\label{symmetry}\sum_{v \in V_1} \sum_{\substack{u \in V_2 \\ \sigma_u \neq \sigma_v}}a_{v,u} = \sum_{u \in V_2} \sum_{\substack{v \in V_1 \\ \sigma_v \neq \sigma_u}}a_{u,v}.\end{equation}
Both equations  \eqref{eqn:payoff_adj} and \eqref{symmetry} will be used in the proof of the main results of this work.
The \textit{social welfare} of a coloring $\sigma$ is defined as the sum of players’ utilities, that is $SW(\sigma) = \sum_{v\in V}\mu_v(\sigma)=2S(\sigma)$.
Moreover, an \textit{optimal strategy profile}, also called \textit{optimal coloring}, is a strategy profile which maximizes $SW(\sigma)$.

For example, consider a graph $G=(V,E)$ as depicted in Figure \ref{1ins}.a, and let assume that the set of colors is $K=\{red, blu, green\}$. The six players are $v_1, v_2, v_3, v_4, v_5, v_6$. Figures \ref{1ins}.a and \ref{1ins}.b shows two different colorings: $\sigma= (red, blue,red,blue,green,red)$ and $\gamma= (green, blue,green,blue,red,red)$,
Let us focus on $\sigma$. The cut of $G$ is $E(\sigma)=\{\{v_1,v_2\},\{v_1,v_4\}, \{v_2,v_3\}, \{v_3, v_4\}, \{v_3, v_5\}, \{v_4,v_5\}, \{v_4, v_6\}, \{v_5,v_6\}\}$, then the cut value is $S(\sigma)=8$. 
The payoff $\mu_1(\sigma) = 2$, since only two of its three neighbors have a color different from its own. Reasoning in this way, we can find the profits of the other five nodes: $\mu_2(\sigma)=2$, 
$\mu_3(\sigma)=3$, $\mu_4(\sigma)=4$, $\mu_5(\sigma)=3$, $\mu_6(\sigma)=2$. 
Therefore, the social welfare of $\sigma$ is $SW(\sigma)=16$.

\begin{figure}[h!]
\begin{center}
\includegraphics[width=12cm]{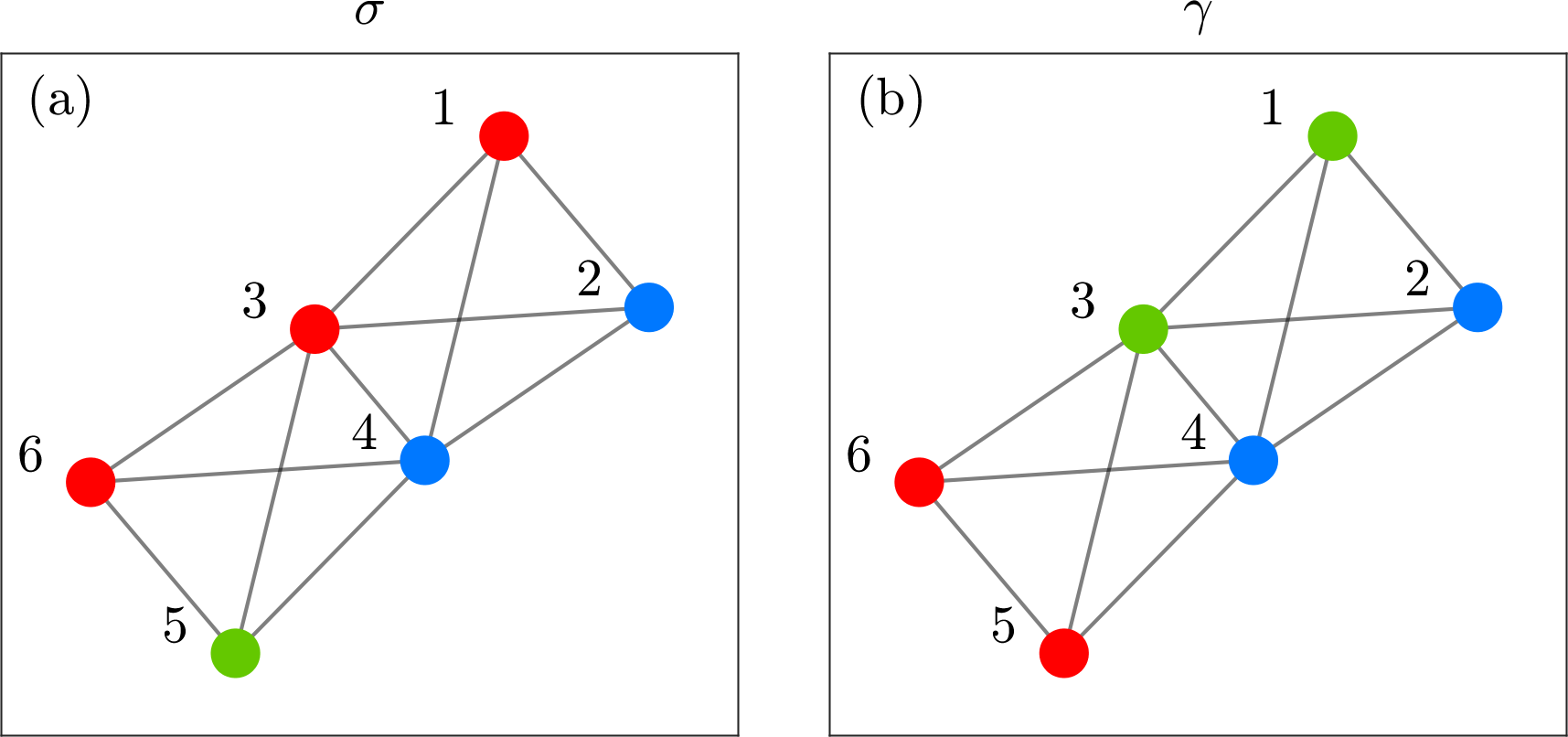}
  \caption{Main notions on colorings and coalitions. In this case, the deviating coalition from $\sigma$ (subplot a) to $\gamma$ (subplot b) is $C = \{v_1, v_3, v_5\}$.}  
  \label{1ins} 
\end{center}
\end{figure}

Given a coalition $C\subseteq V$ and a coloring $\sigma$, first we define the set of colors used by the coalition $C$ in $\sigma$ as $K_C(\sigma)=\{a \in K : ~ \exists v\in C ~ : \sigma_v = a\}$. Then, for each color $a$, we define the set of players in $C$ that have color $a$ in $\sigma$ as $C_{a}(\sigma) = \{v \in C : \sigma_v = a\}$.
We define as $G(C)$ the restriction of the graph $G$ to the members of the coalition $C$, i.e. the set of vertices of $G(C)$ is $C$, while the set of edges is composed by the links connecting only nodes in $C$.

Moreover, if for any player $v\in C$, $\gamma_v \neq\sigma_v$, and for any player $w \notin C$, $\gamma_w=\sigma_w$, then we say that the coalition $C$ \textit{deviates} from the coloring $\sigma$ towards a coloring $\gamma$.
In particular, when a coalition $C$ deviates in such a way that all of its members strictly improve their utility, then we say it performs a \textit{strong improvement}, or \textit{strong deviation}. A strong deviation is said to be \textit{minimal} if no proper subset of the deviating coalition can perform an improvement. The coalition itself is said to be minimal. 
Furthermore, given a coloring $\sigma$, if a coalition $C$ induces a new coloring $\gamma$ after deviating, then we say that the set of edges $E(\sigma)\backslash E(\gamma)$ \textit{enters} the cut, and that the set of edges $E(\sigma)\backslash E(\gamma)$ \textit{leaves} the cut.
For instance, let us consider again the graph $G$ and coloring $\sigma$ reported in Figure \ref{1ins}.a. Let $C$ be the coalition composed of the players $\{v_1,v_3,v_5\}$. The set of colors used by the coalition $C$ in $\sigma_1$ is $K_C(\sigma)=\{red,green\}$. Moreover, the set of players in $C$ that have color $red$ in $\sigma_1$ is $C_{red}(\sigma) = \{v_1,v_3\}$, and the set of the ones that have color $green$ instead is the singleton $C_{green}(\sigma) = \{v_5\}$.
The coloring $\gamma$ reported in Figure \ref{1ins}.b is a deviation from the coloring $\sigma$, which is not strong since the members of coalitions do not improve their utility.

Given a coloring $\sigma$, a player $v$ and a coalition $C$, we denote by $\sigma_{-v}$ and $\sigma_{-C}$ the coloring $\sigma$ besides the strategy played by $v$ and by $C$, respectively. We indicate with $\mu_{C}(\sigma) = \sum_{v \in C} \mu_v(\sigma)$ the total payoff obtained by coalition $C$. Furthermore, $\sigma_C$ denotes the coloring $\sigma$ restricted only to players in $C$.
We say that $\sigma$ is a \textit{Nash Equilibrium} ($NE$) if no player can improve its payoff by deviating unilaterally from $\sigma$, in formulas $\mu_v(\sigma_{-v},a) \leq \mu_v(\sigma)$ for each player $v \in V$ and for each color $a \in K$. 
For each $1 \leq q \leq n$, $\sigma$ is a $q-$\textit{Strong Equilibrium} ($q-SE$) if there exists no coalition $C$ with $|C| \leq q$ that can cooperatively deviate from $\sigma_C$ to $\gamma_C$ in such a way that every player in $C$ strictly improves its utility. 
Therefore, Nash equilibrium (NE) is equivalent to $1-$strong equilibrium. Each optimal equilibrium is also NE \cite{mt}.
We simply call \textit{strong equilibrium} ($SE$) an $n-$strong equilibrium.  

Given two colorings $\gamma$ and $\sigma$, and $C\subseteq V$, we denote
$\displaystyle  P_C(\sigma,\gamma):=\big|\lbrace\{v,j\}\in E:
				v,j\in C,\ \sigma_v=\sigma_j,\ \gamma_v\neq\gamma_j\rbrace\big|$.
Notice that, $P_C(\sigma, \gamma)$ can be rewritten as follows:

\begin{equation}\label{eqn:PC}
P_C(\sigma, \gamma) = \displaystyle\frac{1}{2}\sum_{v \in C} \sum_{\substack{j \in C \\ \gamma_j \neq \gamma_v \\ \sigma_j = \sigma_v}}a_{v,j}.
\end{equation}				

Lastly, we say that $H$ is an \textit{isolated component} of the graph $G$, if $G(H)$ is a connected subgraph of $G$ and, for each vertex $i\in V(H)$, for each node $j\notin V(H)$, $\{i,j\}\notin E$.\\


Now, we report some important results of \cite{mt}, that will be useful in the sequel.   
The first is a proposition that we will use alongside the pigeon-hole principle. The other one is the main result of \cite{mt}, that we are extending in the present work.

\begin{prop}[Prop. 1a from \cite{mt}]
\label{prop1fiorda}
				Let $\sigma$ be a coloring for a graph $G$ and let $C\subseteq V$
				be a minimal subset that can perform a strong deviation from $\sigma$
				to another coloring $\gamma$; then $K_C(\sigma)=K_C(\gamma)$.
			\end{prop}
			

			\begin{theorem}[Thm. 4 from \cite{mt}]
			\label{main_thm_fiorda}
				Any optimal coloring for a max $k$-cut game on
				an unweighted, undirected graph is a 5-Strong equilibrium.
			\end{theorem}


We conclude this section by defining an applicative example that we will cover in our article.
Let us consider the following situation.

\begin{example}\label{example1} An $R\&D$ department of a company must be divided into three teams, each of which will work on a product: however, each of the 12  employees of that department wishes not to work together with some of his colleagues. Figure \ref{fig:fig1}.a reports a graphical representation of the example, where nodes correspond to employees, and edges denote the willingness of connected nodes not to belong to the same team.
Figure \ref{fig:fig1}.b shows  a generic assignment $\alpha$ of employees to the three projects, represented by different colors.
Profit means having the greatest number of people you do not want to work with, on projects other than your own.
Some employees assigned to some team color are assumed to be unhappy (for example nodes $v_5$ and $v_{11}$), since they must work with undesired colleagues. Is it possible to find an optimal allocation, for example the one reported in Figure \ref{fig:fig1}.c, which maximizes the profit by minimizing the number of unhappy people?


This game can also be seen as an adaptation to game theory of the max $k$-cut problem, which consists in partitioning a graph into disjoint sets of nodes in order to minimize the total number of edges between nodes of the same part. Indeed, by optimal coloring, in this case, we intend to minimize the number of pairs of employees who do not want to work together and who are assigned to the same project. To deviate means to change projects.
Strong equilibrium designates the configuration such that if any number of players tried to change the project all together, they would find themselves working on a project with a number of people with whom they would not want to work greater or equal than the previous one.
In this context, then if $\delta(v)$ is the maximum number of people to which individual $v$ is connected, we indicate with $\delta(v,\sigma,a) \leq \delta(v)$ the number of those she do not want to work with who are doing the project $a$.
\end{example}	

The main contribution of this paper is to analyze whether, given a group of unsatisfied people, it is possible to move from an optimal configuration like the $\gamma$ in Figure \ref{fig:fig1}.c to another, such that the payoff of each member of the coalition is increased. 

\begin{figure}
    \centering
    \includegraphics[width=\textwidth]{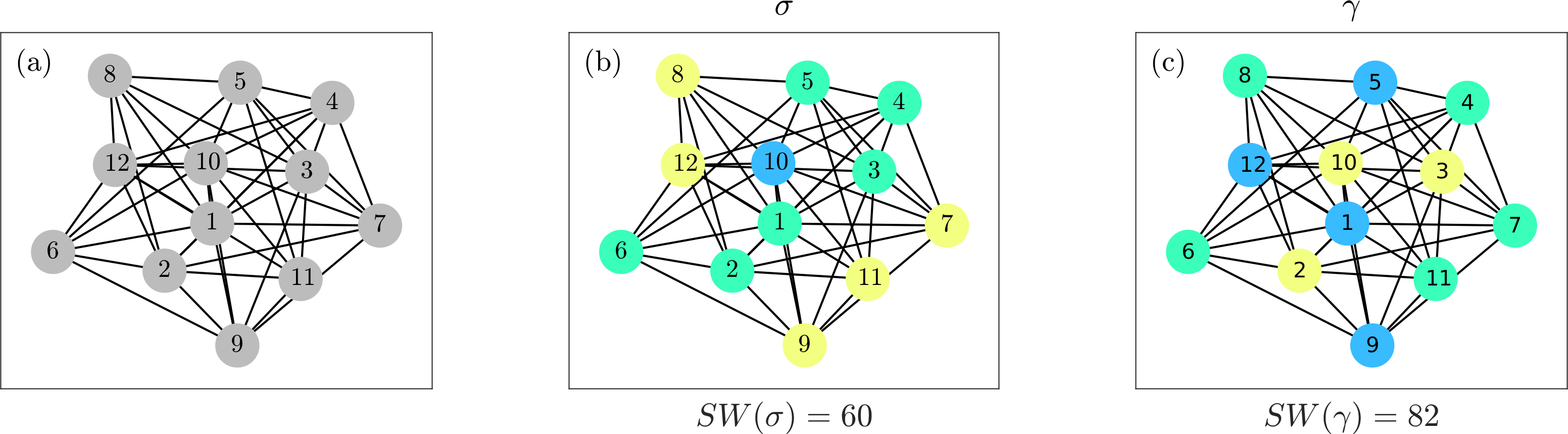}
    \caption{Subplot (a): a network with $n=12$ nodes representing the R\&D department described in Example \ref{example1}.
    Subplot (b): a generic coloring $\sigma$ with $k=3$ colors. 
    Subplot (c): an optimal coloring $\gamma$ with $k=3$ colors.}
    \label{fig:fig1}
\end{figure}

			
\section{The existence of a 7-SE in unweighted graphs}
In \cite{mt} it is shown the existence of a 5-SE in unweighted graphs. We extend this result, by proving the existence of a 7-SE.
We start by showing two propositions which are the core of our main result. 

\begin{lemma}
\label{prop_PI}
Given a graph $G$, let $\sigma$ be a NE for $\mathcal{G}$ and let $C \subseteq V$ strongly deviate from $\sigma$ to $\gamma \in K^n$. Then $\Delta S(\sigma, \gamma) \geq |C| - P_{C}(\sigma, \gamma).$
\begin{proof}
We start by noticing that for both $\gamma$ and $\sigma$ the payoff $S$ can be calculated by considering separately the nodes belonging to $C$ and to $V\setminus C $:

$$
\begin{array}{rcl}
S(\sigma) & = & \displaystyle\frac{1}{2} \left[\sum_{v \in C}\left(\sum_{\substack{w \in C \\ \sigma_w \neq \sigma_v}}a_{v,w} + \sum_{\substack{w \in V \setminus C \\ \sigma_w \neq \sigma_v}}a_{v,w} \right) + \sum_{v \in V \setminus C}\left(\sum_{\substack{w \in C \\ \sigma_w \neq \sigma_v}}a_{v,w} + \sum_{\substack{w \in V \setminus C \\ \sigma_w \neq \sigma_v}}a_{v,w} \right)\right],\\
\end{array}
$$

and

$$
\begin{array}{rcl}
S(\gamma) & = & \displaystyle\frac{1}{2} \left[\sum_{v \in C}\left(\sum_{\substack{w \in C \\ \gamma_w \neq \gamma_v}}a_{v,w} + \sum_{\substack{w \in V \setminus C \\ \gamma_w \neq \gamma_v}}a_{v,w} \right) + \sum_{v \in V \setminus C}\left(\sum_{\substack{w \in C \\ \gamma_w \neq \gamma_v}}a_{v,w} + \sum_{\substack{w \in V \setminus C \\ \gamma_w \neq \gamma_v}}a_{v,w} \right)\right],\\
\end{array}
$$

Since the color of nodes in $V \setminus C$ does not change, we have that:

$$\sum_{v \in V \setminus C}\sum_{\substack{w \in V \setminus C \\ \sigma_w \neq \sigma_v}}a_{v,w} = \sum_{v \in V \setminus C}\sum_{\substack{w \in V \setminus C \\ \gamma_w \neq \gamma_v}}a_{v,w}.$$

Additionally, using equation \eqref{symmetry} in the Preliminaries section, we have that:

$$\begin{array}{rcl}
\Delta S(\sigma, \gamma) & = & S(\gamma) - S(\sigma)\\
& = & \displaystyle\frac{1}{2}\sum_{v \in C} \left[\left(\sum_{\substack{w \in C \\ \gamma_w \neq \gamma_v}}a_{v,w} + 2\sum_{\substack{w \in V \setminus C \\ \gamma_w \neq \gamma_v}}a_{v,w} \right) - \left(\sum_{\substack{w \in C \\ \sigma_w \neq \sigma_v}}a_{v,w} + 2\sum_{\substack{w \in V \setminus C \\ \sigma_w \neq \sigma_v}}a_{v,w} \right)\right] \\
& = & \displaystyle\frac{1}{2}\sum_{v \in C} \left[\left(2\sum_{\substack{w \in C \\ \gamma_w \neq \gamma_v}}a_{v,w} + 2\sum_{\substack{w \in V \setminus C \\ \gamma_w \neq \gamma_v}}a_{v,w} \right) - \left(2\sum_{\substack{w \in C \\ \sigma_w \neq \sigma_v}}a_{v,w} + 2\sum_{\substack{w \in V \setminus C \\ \sigma_w \neq \sigma_v}}a_{v,w} \right)\right] \\
& & - \displaystyle\frac{1}{2}\sum_{v \in C} \left(\sum_{\substack{w \in C \\ \gamma_w \neq \gamma_v}}a_{v,w}  - \sum_{\substack{w \in C \\ \sigma_w \neq \sigma_v}}a_{v,w} \right).
\end{array}
$$

Moreover, using equation \eqref{eqn:payoff_adj}, we have that:
$$
\begin{array}{rcl}
\Delta S (\sigma, \gamma)& = & \displaystyle\sum_{v \in C} (\mu_v(\gamma) -\mu_v(\sigma))- \displaystyle\frac{1}{2}\sum_{v \in C} \left(\sum_{\substack{w \in C \\ \gamma_w \neq \gamma_v}}a_{v,w}  - \sum_{\substack{w \in C \\ \sigma_w \neq \sigma_v}}a_{v,w} \right) \\
& = & \displaystyle\sum_{v \in C} (\mu_v(\gamma) -\mu_v(\sigma))- \displaystyle\frac{1}{2}\sum_{v \in C} \left(\sum_{\substack{w \in C \\ \gamma_w \neq \gamma_v \\ \sigma_w = \sigma_v}}a_{v,w} + \sum_{\substack{w \in C \\ \gamma_w \neq \gamma_v \\ \sigma_w \neq \sigma_v}}a_{v,w}  - \sum_{\substack{w \in C \\ \sigma_w \neq \sigma_v \\ \gamma_w = \gamma_v}}a_{v,w} - \sum_{\substack{w \in C \\ \sigma_w \neq \sigma_v \\ \gamma_w \neq \gamma_v}}a_{v,w} \right) \\
& = & \displaystyle\sum_{v \in C} (\mu_v(\gamma) -\mu_v(\sigma))- \displaystyle\frac{1}{2}\sum_{v \in C} \left(\sum_{\substack{w \in C \\ \gamma_w \neq \gamma_v \\ \sigma_w = \sigma_v}}a_{v,w}  - \sum_{\substack{w \in C \\ \sigma_w \neq \sigma_v \\ \gamma_w = \gamma_v}}a_{v,w} \right).
\end{array}$$

Using equation \eqref{eqn:PC}, we finally get that:

$$\Delta S( \sigma, \gamma) =  \displaystyle\sum_{v \in C} (\mu_v(\gamma) -\mu_v(\sigma))-P_{C}(\sigma, \gamma) + P_{C}(\gamma, \sigma).
$$

On the other hand, since $C$ is a strong deviation, then:

$$\mu_v(\gamma) \geq \mu_v(\sigma) + 1 ~\forall v \in C.$$

This yields to:

$$\displaystyle\sum_{v \in C} (\mu_v(\gamma) -\mu_v(\sigma)) \geq \displaystyle\sum_{v \in C} 1 = |C|.$$

Hence:
$$\begin{array}{rcl}
\Delta S(\sigma, \gamma) & \geq & |C| -P_{C}(\sigma, \gamma) + P_{C}(\gamma, \sigma) \\ 
& \geq & |C| -P_{C}(\sigma, \gamma).
\end{array}
$$
The last inequality holds since $P_{C}(\gamma, \sigma) \geq 0$.

\end{proof}
\end{lemma}


The previous Proposition will be used in the proof of the next theoretical results.

\begin{prop}
\label{prop3_alt}
Given a graph $G$, let $\sigma$ be a $NE$ for it and let $C\subseteq V$ be a minimal coalition that strongly deviates from a coloring $\sigma$ to a different one $\gamma$.
If $|K_C(\sigma)| \in \{|C|-3, |C|-2\}$, then $\Delta S(\sigma,\gamma)>0$. 
\begin{proof}
For each feasible dimension of $|C|$,
we consider the maximum $P_C(\sigma, \gamma)$ in each configuration.

\begin{itemize}
    \item \textbf{Case a: $|K_C(\sigma)| = |C|-2$}\\
In this situation, for the Pigeon-hole principle, we have the following subcases:

\begin{itemize}
    \item[a.1)] $|C_a(\sigma)| = 3$, $|C_b(\sigma)| =1 ~\forall b \in K_C(\sigma) \setminus \{a\}$. Assuming at least two colors in $C$, we have that $|C| \geq 4$.

    \item[a.2)] $|C_a(\sigma)| = |C_b(\sigma)| = 2$, $|C_c(\sigma)| =1 ~\forall c \in K_C(\sigma) \setminus \{a,b\}$. 
    The presence of singleton colors $c$ is not required.
    Also in this case, $|C| \geq 4$ guarantees the presence of at least two colors.
\end{itemize}

\item \textbf{Case b: $|K_C(\sigma)| = |C|-3$}\\
In this situation, for the Pigeon-hole principle, we have the following subcases:

\begin{itemize}
    \item[b.1)] $|C_a(\sigma)| = 4$, $|C_b(\sigma)| =1 ~\forall b \in K_C(\sigma) \setminus \{a\}$. For this case, $|C| \geq 5$ in order to have at least two colors.
    
    \item[b.2)] $|C_a(\sigma)| = 3$, $|C_b(\sigma)| = 2$, $|C_c(\sigma)| =1 ~\forall c \in K_C(\sigma) \setminus \{a,b\}$.
    Here, the presence of singleton colors is not required, and $|C| \geq 5$ in order to have at least two colors.
    \item[b.3)] $|C_a(\sigma)| = |C_b(\sigma)| = |C_c(\sigma)| = 2, 
    |C_d(\sigma)| = 1  ~\forall d \in K_C(\sigma) \setminus \{a, b, c\}$.
    The singleton colors are not required.
    In this case, $|C| \geq 6$ and the number of colors it at least $3$.
    
\end{itemize}
\end{itemize}

For each case and for each feasible value of $|C|$, we evaluate the maximum value of $P_C(\sigma, \gamma)$.
This evaluation is graphically depicted in Figure \ref{fig:prop3}, where $\gamma$ colorings maximize the number of edges connecting nodes with different colors, which were sharing the same colors in $\sigma$ colorings.
Notice that, since $C$ is minimal, all colors in $\gamma$ are the same as in $\sigma$. This is coherent with the Proposition \ref{prop1fiorda}.
The maximum value of $P_C(\sigma, \gamma)$, as well as the minimum value of $\Delta S (\sigma, \gamma)$ are calculated in Table \ref{tab:prop3} according to Lemma \ref{prop_PI}, i.e. $\Delta S(\sigma, \gamma) \geq |C| - P_C(\sigma, \gamma)$. \\

As shown in Table \ref{tab:prop3}, $\Delta S(\sigma, \gamma) > 0$ in all cases.

\begin{figure}[ht]
    \centering
    \includegraphics[width=\textwidth]{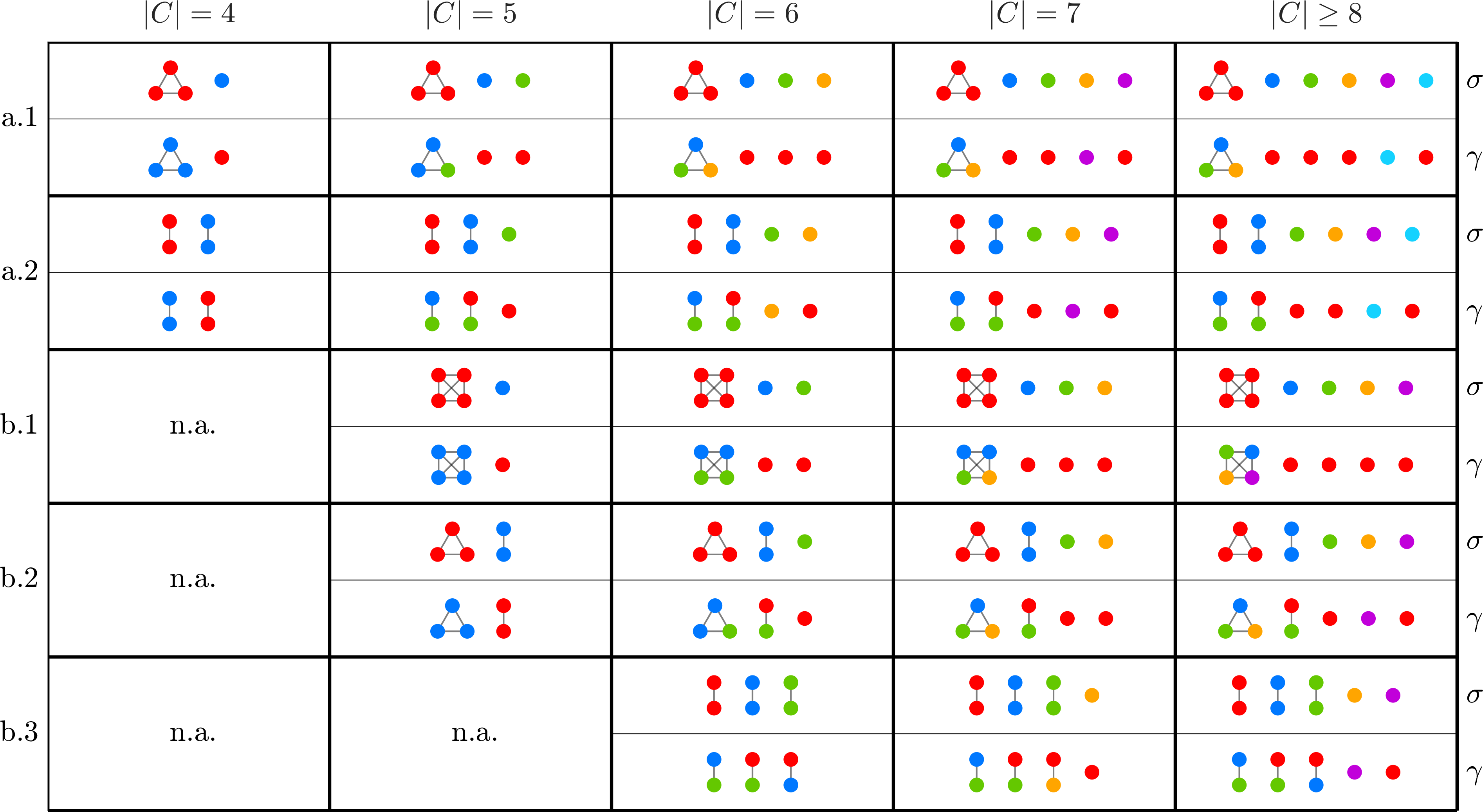}
    \caption{Graphical representation of the configurations analysed in Proposition \ref{prop3_alt}. N.a. entries indicate non feasible cases. }
    \label{fig:prop3}
\end{figure}

 \begin{table}[ht]
    \centering
\begin{tabular}{|c|c|c|c|c|c|c|}
\hline
\multicolumn{2}{|c|}{~} & $|C| = 4$ & $|C| = 5$ & $|C| = 6$ & $|C| = 7$ & $|C| \geq 8$ \\
 \hline %
  \hline %
\multirow{2}{*}{a.1} & $\max  P_C(\sigma,\gamma)$ &  $0$ & $2$ & $3$ & $3$ & $3$\\
\cline{2-7}
 & $ \Delta S(\sigma, \gamma)$ &  $4$ & $3$ & $3$ & $4$ & $\geq 5$\\
\hline\hline
\multirow{2}{*}{a.2} & $\max P_C(\sigma,\gamma)$ &  $0$ & $2$ & $2$ & $2$ & $2$\\
\cline{2-7}
 & $ \Delta S(\sigma, \gamma)$ &  $4$ & $3$ & $4$ & $5$ & $\geq 6$\\
\hline\hline
\multirow{2}{*}{b.1} & $\max P_C(\sigma,\gamma)$ &  n.a. & $0$ & $4$ & $5$ & $6$\\
\cline{2-7}
 & $ \Delta S(\sigma, \gamma)$ &  n.a. & $5$ & $2$ & $2$ & $\geq 2$\\
\hline\hline
\multirow{2}{*}{b.2} & $\max P_C(\sigma,\gamma)$ &  n.a. & $0$ & $3$ & $4$ & $4$\\
\cline{2-7}
 & $ \Delta S(\sigma, \gamma)$ &  n.a. & $5$ & $3$ & $3$ & $\geq 4$\\
\hline\hline
 \multirow{2}{*}{b.3} & $\max P_C(\sigma,\gamma)$ &  n.a. & n.a. & $3$ & $3$ & $3$\\
\cline{2-7}
 & $ \Delta S(\sigma, \gamma)$ &  n.a. & n.a. & $3$ & $4$ & $\geq 5$\\
\hline
\end{tabular}
    \caption{Evaluation of $\Delta S(\sigma, \gamma)$ for the the different configurations  analysed in Proposition \ref{prop3_alt}, depending on $|C|$ and $P_C(\sigma,\gamma)$. N.a. entries indicate non feasible cases.}
     \label{tab:prop3}
 \end{table}

\end{proof}

\end{prop}

The following proposition shows that in the coalition there is a rearrangement of colors due to the strong deviation such that each vertex assumes the color of another vertex within the coalition itself.

\begin{prop}\label{prop4}
	Given a graph $G=(V,E)$, let $\sigma$ be a NE for $\mathcal{G}$ and let $C\subseteq V$
			be a minimal subset w.r.t. strongly deviating from a coloring $\sigma$
			to $\gamma$, with $\sigma \neq \gamma$. Then, for each vertex $v \in C$,
			there exists a vertex $w \in C$ such that $v\neq w$, $\sigma_v \neq \sigma_w$, $\{v,w\}\in E$
			and $\gamma_v=\sigma_w$.
		\end{prop}
	\begin{proof}
The proof will be split into two halves: 
\begin{enumerate}
\item First, we show that each vertex in $C$ deviates towards the color that one of its neighbors has in $\sigma$. 
\item Then, we prove that such neighbor must be in $C$, meaning that each vertex deviates to a color used in $C$.
\end{enumerate}
Concerning part 1, we show that for each vertex $u \in C$, it holds that $\delta(u,\sigma, \gamma_u)>0$. Indeed, let assume by contradiction that there exists a vertex $t\in C$ such that $\delta(t,\sigma,\gamma_{t})=0$. 
We distinguish the two sub-cases:
		
		\begin{enumerate}
			\item If $\delta(t,\sigma,\sigma_{t})=0$,
		then $t$ does not need to improve its payoff
			as it is already earning its maximum payoff; but this contradicts
			the hypothesis that $C$ is a strong deviation.
			
			\item If $\delta(t,\sigma,\sigma_{t})>0$, since $\delta(t,\sigma,\gamma_{t})=0$, we get that $\{t\}$ is a minimal strong deviation from $\sigma$. However, this contradicts the hypothesis that $\sigma$ is a NE and that $C$
			is minimal w.r.t. strongly deviating from $\sigma$.
		\end{enumerate}
Therefore, for each vertex $u\in C,\, \delta(u,\sigma, \gamma_u)>0$. Notice that this fact holds also if $u$ deviates towards the color of another vertex in $C$ that is not one its neighbors.

Concerning part 2, let $u\in C$ be such that $\gamma_{u}=\sigma_t$,
		for some $t\in V\backslash C$, and $\{s,t\}\in E$ with $\sigma_s\neq\sigma_t$, for each node $s\in C$.

Since $\sigma$ is a NE, it holds that $\delta(u,\sigma,\sigma_t)\ge
\delta(u,\sigma,\sigma_{u})$.

Moreover, we are assuming that $C$ strongly deviates, we get that $\mu_{u}(\gamma)>\mu_{u}(\sigma)$,
		which means that $\delta(u,\gamma, \sigma_t)<
		\delta(u, \sigma, \sigma_{u})$.
Hence, $\delta(u,\sigma,\sigma_t) > \delta(u,\gamma,\sigma_t).$

Anyway, it might happen that some neighbors of $u$ in $\gamma$ deviate towards $\sigma_t$, meaning that $\delta(u, \gamma, \sigma_t)\geq \delta(u, \sigma, \sigma_t)$, thus obtaining a contradiction.

Therefore, for each vertex $v\in C$,
		there must be a vertex $w\in C$ such that $w\neq v$, $\{v,w\}\in E$ and
		$\gamma_v=\sigma_w$. This means that a vertex $u \in C$ must deviate to the color of one of its neighbors in $C$.
	\end{proof}
	
In our example, Proposition \ref{prop4} asserts that one member of the coalition is so disliked by another that he leaves the project to him.
	
\begin{lemma}\label{lemma:ne_singleton}
Let $G=(V,E)$ be a graph,  $\sigma$ a NE for $G=(V,E)$ and $C\subseteq V$ a strongly deviating coalition from a coloring $\sigma$
		to $\gamma$, with $\sigma \neq \gamma$. Then, $|C|\geq  2$.
\end{lemma}	
\begin{proof}
Suppose that $C=\{v\}$ for a node $v \in V$. If $C$ is a strong deviation, then $v$ is able to change its color in order to improve unilaterally its payoff, namely $\mu_v(\gamma) > \mu_v(\sigma)$. This contradicts the fact that $\sigma$ is NE. 
\end{proof}

The following proposition allows not to consider non-connected subgraphs in the proof of Proposition \ref{prop3_alt}.
	\begin{prop}\label{prop_iso}
Let $G$ be a graph, $K$ the set of all possible colors, $\sigma$ a NE for $G$ and $C\subseteq V$ a minimal subset w.r.t. strongly deviating from a coloring $\sigma$
		to $\gamma$, with $\sigma \neq \gamma$. Then, $G(C)$ is an isolated component of $G$.
	\end{prop}
	\begin{proof}
First, we show that no isolated component of $G(C)$ can be monochromatic, i.e. every isolated component $H = (V_H, E_H)$ of $G(C)$ is such that there is no color $a$ in $K$ such that for every vertex $v$ in $V_H$, $\sigma_v = a$. In other words, we will show that each isolated component of $G(C)$ must have representatives of at least two colors. 
	
From Lemma \ref{lemma:ne_singleton}, $|V_H| \geq 2$.
Assuming by absurd that all members of $V_H$ have the same color $a$, we have that in the profitable deviation each vertex of $V_H$ helps the other vertices of $V_H$ to increase their profit. This means that each vertex of $V_H$ would make a greater profit by deviating itself towards $\gamma$, contradicting the minimality of $C$.
Now, we are able to show that under the considered assumptions, $G(C)$ is an isolated component. We reason by contrapositive. Suppose that the number of isolated components of C is greater than $1$. We have just shown that none of these isolated components can be monochromatic. Let us consider two of these isolated components $C_1$ and $C_2$ and denote by $d(u, v)$ the length of the shortest path between the vertices $u \in C_1$ and $v \in C_2$, in terms of edges. If $\min_{u \in C_1, v \in C_2} d(u, v) \geq 1$, then the shortest path between any vertex of $C_1$ and any vertex of $C_2$ passes through at least one vertex of $V \backslash C$. Since by the Proposition \ref{prop3_alt}, each vertex of $C$ can only deviate towards the color of its neighbor, a vertex of $C_1$ deviates towards the color in $\sigma$ of another vertex of $C_1$. The same situation happens for $C_2$. Hence, $C_1$ does not need $C_2$ to deviate, and vice versa. This contradicts the minimality of $C$ in deviating strongly.
Therefore, $G(C)$ must be an isolated component of $G$.
\end{proof}
	
Now, we prove a Theorem stating that for all the cases where a minimal strongly deviating coalition $C$ has at most 7 vertices, the cut value increases. 
	\begin{theorem}
	\label{teo4}
		Let $\sigma$ be a NE for $\mathcal{G}$ and $C\subseteq V$
		be a minimal subset w.r.t. strongly deviating from a coloring $\sigma$
		to another $\gamma$. If $|C|\leq7$,
		then $\Delta S (\sigma,\gamma)>0$.
	\end{theorem}
	\begin{proof}
In order to prove the Theorem, we need to consider all the cases for which $|C| \leq 7$. Notice that the cases $|C| \leq 5$ are already covered by \cite{mt}. 
Additionally, in Proposition 4 of \cite{mt}, it has been proven that $C$ is a minimal strongly deviating coalition for $|K_C(\sigma)| \in \{2, 5, 6\}$ with $|C| = 6$, and $|K_C(\sigma)| \in \{2, 6, 7\}$ with $|C| = 7$.
From Proposition \ref{prop3_alt}, we know that in the cases $|C|=6$ and $|C|=7$, if $|K_C(\sigma)| \geq 6 - 3 = 3$ or $|K_C(\sigma)| \geq 7 - 3 = 4$, then $\Delta S(\sigma, \gamma) > 0$.  
Then, only the case $|K_C(\sigma)| = 3$ must be accounted.
Let assume that $K_C(\sigma)=\{a,b,c\}$.
The only possible configurations are the following four:
\begin{enumerate}
\item[\textit{Configuration 1}.] $|C_{a}(\sigma)|=3$, $|C_{b}(\sigma)|=|C_{c}(\sigma)|=2$.
\item[\textit{Configuration 2}.] $|C_{a}(\sigma)|=|C_{b}(\sigma)|=3$, $|C_{c}(\sigma)|=1$.
\item[\textit{Configuration 3}.] $|C_{a}(\sigma)|=4$, $|C_{b}(\sigma)|=2$, and $|C_{c}(\sigma)|=1$.
\item[\textit{Configuration 4}.] $|C_{a}(\sigma)|=5$, $|C_{b}(\sigma)|=|C_{c}(\sigma)|=1$.
\end{enumerate}
These configuration are briefly reported in Figure \ref{fig:conf7}.
\begin{figure}[ht]
    \centering
    \includegraphics[scale=0.7]{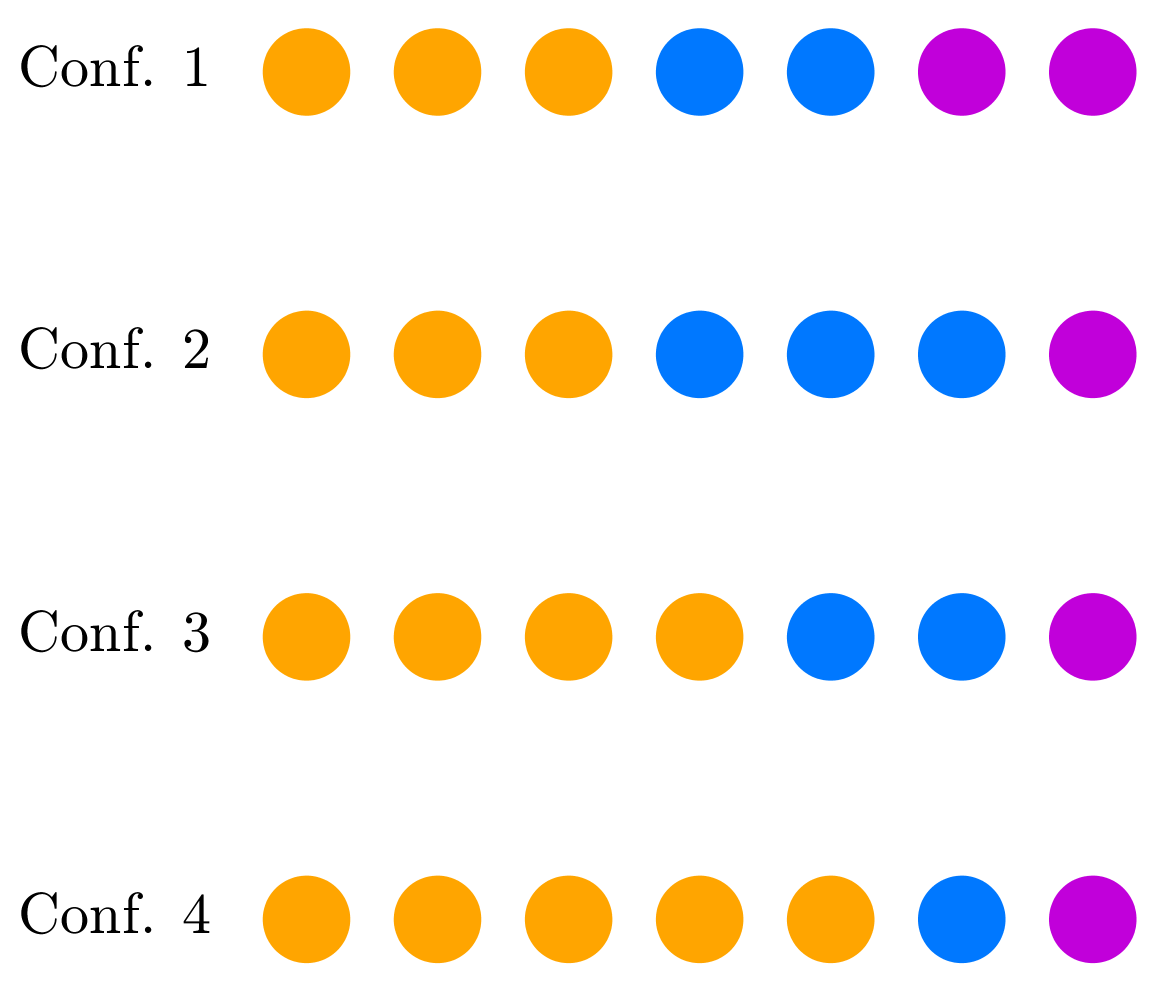}
    \caption{Pictorial representation of the coalition colorings covered by Theorem \ref{teo4}.}
    \label{fig:conf7}
\end{figure}
Concerning configurations $1$ and $2$, the maximum value reachable by $P_C(\sigma, \gamma)$ is less than 7; indeed
we have that $P_C(\sigma, \gamma)=5$ for configuration a and $P_C(\sigma, \gamma) = 6$ in configuration b. Therefore, $\Delta S (\sigma, \gamma) \geq 1$.
Concerning configuration $3$, we have that $P_C(\sigma, \gamma) = 7$ if and only if each of the four vertices of color $a$ in $\sigma$ deviates towards a different color from the other three. It can easily be seen that at least two other colors and thus vertices are needed to guarantee this deviation, even if a vertex of color $b$ and one of color $a$ take on each other's colors and that the remaining vertex of color $b$ deviates together with another vertex of color $a$ towards the color $c$. Therefore, $P_C(\sigma, \gamma)<7$ and $\Delta S(\sigma, \gamma) \geq 1$. 
Finally, consider configuration 4. All the feasible values that $P_C(\sigma, \gamma)$ can assume are:
\begin{itemize}
\item $P_C(\sigma, \gamma)=10$. We reason similarly to configuration 3, obtaining that the vertices of color $a$ need to change towards five different colors, while only $2$ colors are available for them. Therefore, if $ | C | =  7$ it cannot happen that $ P_C(\sigma, \gamma) =  10$.
\item $P_C(\sigma, \gamma)=9$. We reason similarly to configuration 3 when $ P_C(\sigma, \gamma) = 6 $, and we obtain that the five vertices of color $ a $ need to change towards either four or five different colors to be able to deviate, while only $2$ colors are available for them.
\item $P_C(\sigma, \gamma)=8$. We reason similarly to configuration 3 when $ P_C(\sigma, \gamma) = 6 $, and we obtain that the five vertices of color $a$ need either three or four or five different colors to be able to deviate, while only $2$ colors are available for them.
\item $P_C(\sigma, \gamma)=7$. We observe that it is not possible that $ C_{a} $ indicates a subgraph without loops of length 3. Indeed, all possible graphs of $5$ nodes and $7$ edges are reported in Figure \ref{fig:loop3}. The presence of a loop of length 3 and the fact that only two colors are available, implies that $P_C(\sigma, \gamma) \neq 7$.
\end{itemize}
Then $P_C(\sigma, \gamma) \leq 6$, and
therefore, in any case $ \Delta S (\sigma, \gamma)> 0 $, i.e. $ S (\gamma)> S (\sigma) $.

	\end{proof}
	
\begin{figure}
    \centering
    \includegraphics[scale=.7]{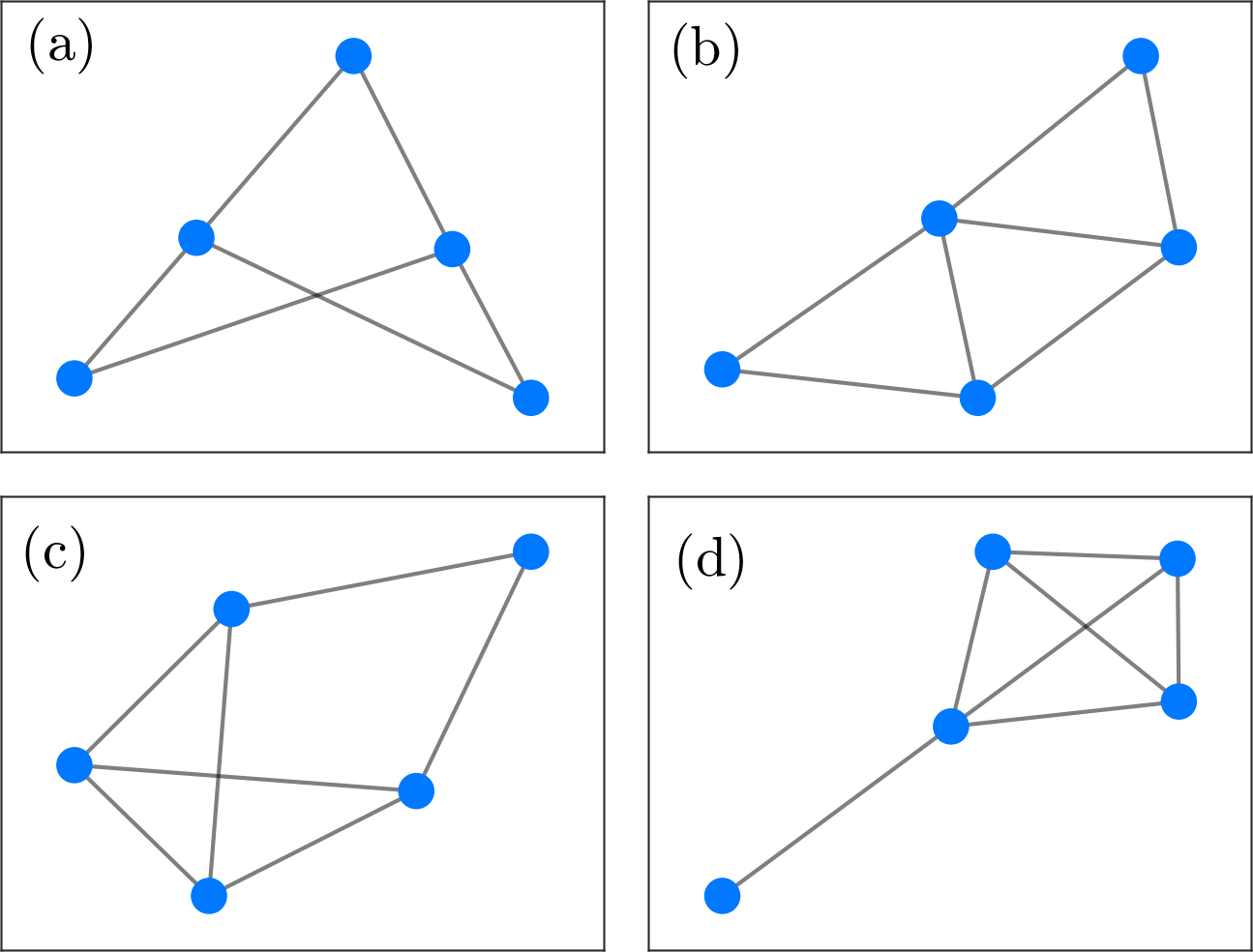}
    \caption{All possible ways in which $5$ nodes can be connected according to configuration 4 of Theorem \ref{teo4} for $P_C(\sigma, \gamma) = 7$.}
    \label{fig:loop3}
\end{figure}
With reference to Example \ref{example1}, Theorem 2 can be read as follows: if a group has more than seven employees then the number of pairs of employees who do not want to work together and who participate in the same project is reduced.
	
We are now ready to prove the main result of the paper:
	\begin{theorem}\label{thm:final_result}
	\label{teo5}
		Let $\sigma$ be an optimal coloring for $\mathcal{G}$.
		Then $\sigma$ is also a 7-Strong equilibrium.
	\end{theorem}
	\begin{proof}
	Assume that the optimal coloring is not a 7-SE. From Theorem \ref{teo4} we can easily conclude that this would lead to a contradiction. 
%
%
%
	\end{proof}

In our example, Theorem \ref{thm:final_result} can be read as follows: if management allocates employees to projects in order to minimize the number of pairs of people who do not want to work together and who are assigned to the same project, then groups of at most seven employees are unable to apply for a transfer to other projects to minimize the number of undesired collaborations.
Therefore, it can be seen that the seven employees have no way to move to other projects without having to work alongside people they do not want to work with.	
Figure \ref{fig:fig2} depicts all the optimal colorings available for the network introduced in Example \ref{example1}. In particular, we assign them the following names: $\sigma$ is the benchmark case, reported in Figure \ref {fig:fig1}.c and in Figure \ref{fig:fig2}.a, $\gamma^1$, $\gamma^2$, $\gamma^3$, $\gamma^4$ and $\gamma^5$ are all the other optimal configurations, shown in Figures \ref{fig:fig2}.b-\ref{fig:fig2}.f. The last five optimal colorings have been compared to $\sigma$, and the corresponding deviating coalition $C$ between each couple of optimal colorings has been computed in terms of deviating colors.  Specifically, the members of the coalition $\sigma$ that deviate to another optimal coloring $\gamma^i$, with $i \in \{1, \ldots, 5\}$ are highlighted by bold black borders. We notice that in all cases, the payoffs $\mu_C(\gamma^i)$ collected by the members of these coalitions are equal to $\mu_C(\sigma)$, the one obtained in the benchmark $\sigma$. This result shows that no strong deviating coalitions of $7$ (or more) members exist in the considered game. Additionally, we also verified that $\mu_v(\sigma) = \mu_v(\gamma^i) ~\forall v \in C$ and for $i \in \{1, \ldots, 5\}$.

\begin{remark}
Note that in Theorem \ref{teo5}, we only use the fact that any optimal coloring is resistant to strong deviation done by minimal subsets of at most 6 nodes. 	
	The same result also holds for non-minimal subsets (w.r.t. a strong
	deviation) containing at most 8 nodes. Indeed, a non-minimal subset of 8 nodes
	that deviates from an optimal coloring would have a minimal subset deviating
	from the same coloring and containing at most 7 nodes, and this is not possible.
	\end{remark}
	
\begin{figure}
    \centering
    \includegraphics[width=\textwidth]{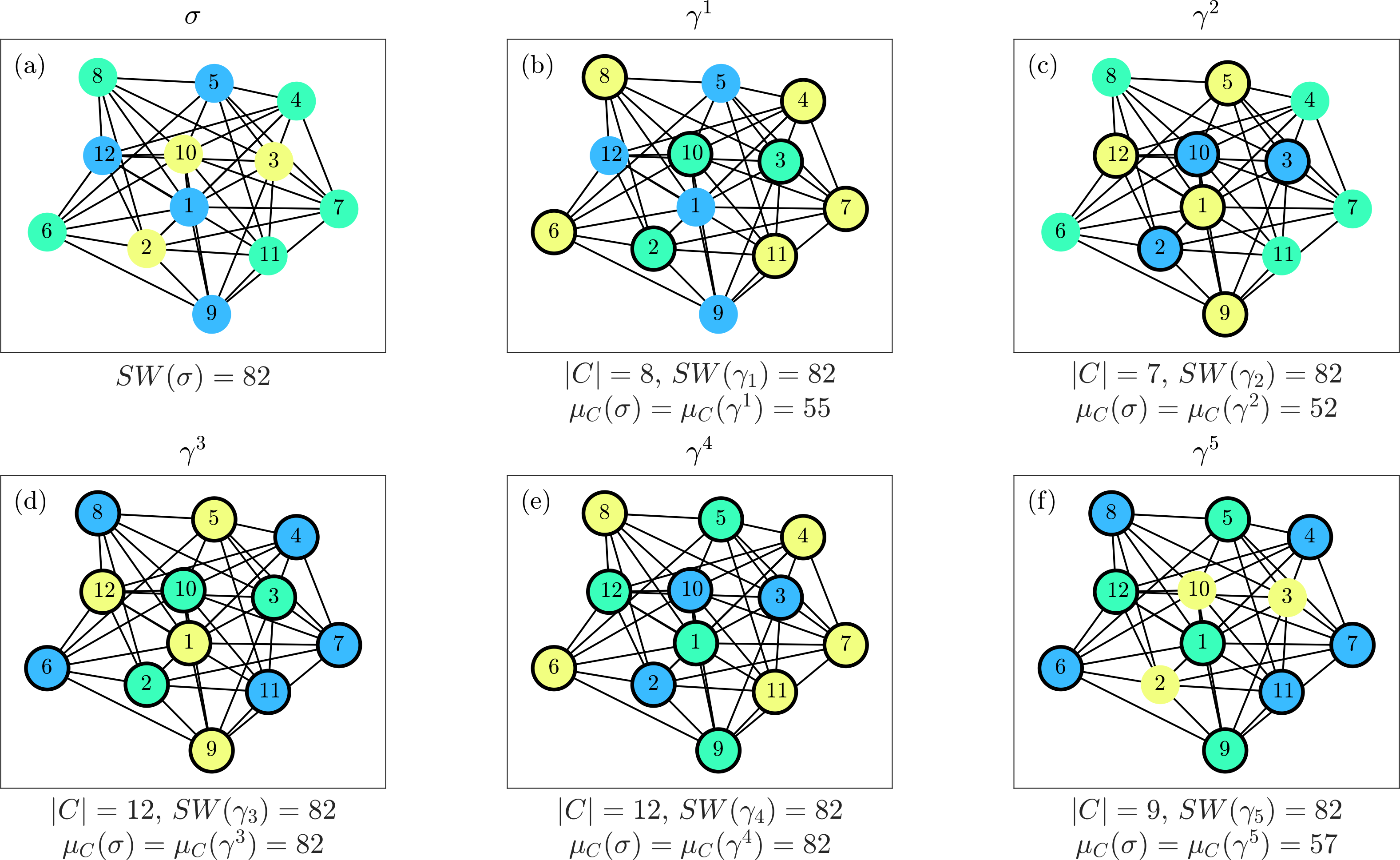}
    \caption{Comparison of the six optimal colorings ($\sigma$, $\gamma^1$, $\gamma^2$, $\gamma^3$, $\gamma^4$ and $\gamma^5$) in a network with $12$ nodes  and $3$ colors. 
    Subplot (a): optimal coloring $\sigma$.
    Subplots (b)-(f): comparison of the optimal colorings  $\sigma$,   $\gamma^1$, $\gamma^2$, $\gamma^3$, $\gamma^4$ and $\gamma^5$. Nodes with black edges represent the deviating coalitions from $\sigma$.}
    \label{fig:fig2}
\end{figure}	
	
\begin{figure}
    \centering
    \includegraphics[width=0.7\textwidth]{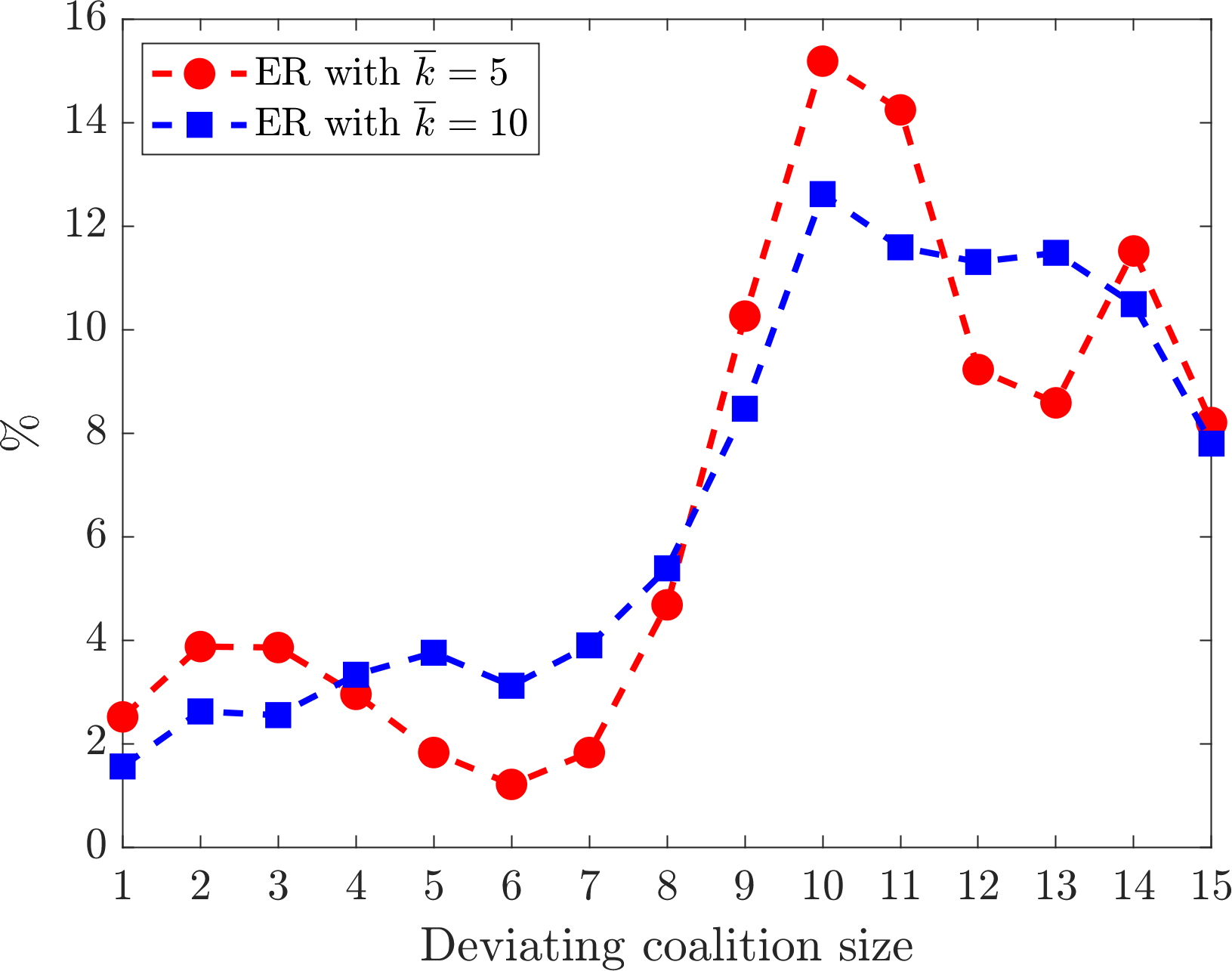}
    \caption{Distribution of the size of the deviating coalitions for different average degrees of $n=15$ Erd\"os-Reny\`i random graphs.}
    \label{fig:fig3}
\end{figure}

\section{Conclusions and future developments}

In this paper, some results on the robustness to $7$-deviating coalitions of optimal colorings is proposed. The approach is grounded on game theory, thus making us able to use the concept of Nash equilibrium of games for discussing the presence of coalitions diverging from optimal configurations.  

The main result of the study is Theorem 3, which proves that optimal colorings of the graph are 7-{\it Strong equilibria}. This means that in order the system is able to switch from one configuration to another, a coalition of more than $7$ diverging elements must be found. Anyway, smaller coalitions of optimal colorings may deviate towards other existing optimal solutions. 

In order to extend the present findings to possibly larger coalitions, thus making the optimal coloring strong enough to prevent divergent solutions, we formulate a conjecture, holding for groups of nodes composed by any number of elements. 

Figure \ref{fig:fig3} motivates the  conjecture. A numerical experiment, considering two groups of Erd\"os-Reny\`i random graphs $10$ of $n=15$ nodes with average degrees $5$ and $10,$ has been developed. The distribution of the optimal colorings of $10$ run for each of the 2 cases are reported. One can notice that optimal colorings of any size are present. Moreover, none of them is strong, since the payoff of deviating  coalitions is never increased. 

	\begin{conjecture}
	Each optimal coloring is {\it SE}.
	\end{conjecture}
	
	To prove Conjecture 1, it is necessary
	to find out more about the properties of optimal colorings. 
	For example, one could study how the neighborhood of a deviating subset
	is colored in an optimal coloring or how clustered colors are in an optimal	coloring.
	
	Another approach is to study the graph induced by nodes
	whose utility is lower than its degree.
	
	Alternatively, one could try to exploit defective colorings \cite{Defcolor1}:
	a ($ k,m $) \textit{defective coloring} (or ($ k,m $)-coloring) for a graph $G$
	is a coloring of $G$ with $k$ colors such that each node has at most $m$
	neighbors of the same color as itself.
	We conjecture that any optimal coloring that is a ($ k $,1)-coloring (respectively,
	($ k $,2)-coloring) for a ($ k $,1)-colorable graph that is not $k$-chromatic
	(respectively, for a ($ k $,2)-colorable graph that is not ($k$,1)-colorable)
	is a strong equilibrium. If it were so, it would also imply that Conjecture 1 is
	true for both planar graphs and toroidal graphs \cite{Defcolor2}.

	Lastly, one could study how strong an optimal coloring is in other games similar
	to the max $k$-cut game, like the generalized graph $k$-coloring games
	\cite{generalized} or the ones from \cite{all-or-nothing} and \cite{gourves-monnot-3}.
	
	

\medskip
Received xxxx 20xx; revised xxxx 20xx.
\medskip

\end{document}